\documentclass{amsart}
\usepackage[usenames, dvipsnames]{color}

\title{$W_m$-algebras and fractional powers of difference operators}
\author{Gloria Mar\'i Beffa}
\address{Mathematics Department\\ University of Wisconsin\\ Madison WI 53706}
\thanks{The author gratefully acknowledge support through research funding from the College of Letters \& Science at UW-Madison}
\usepackage{amsmath,amsfonts,amsthm,amssymb,relsize,mathdots,MnSymbol}

\newtheorem{theorem}{Theorem}
\numberwithin{theorem}{section}
\newtheorem{corollary}[theorem]{Corollary}

\newtheorem{proposition}[theorem]{Proposition}

\theoremstyle{definition}

\def\deltaF{\delta_{D}\F}
\def\deltaG{\delta_{D}\G}

\def\RR{\mathbb R}
\def\RP{\mathbb{RP}}
\def\Z{\mathbb Z}

\def\SL{\mathrm {SL}}
\def\GL{\mathrm {GL}}
\def\PSL{\mathrm {PSL}}

\def\tr{\mathrm{Tr}}

\def\T{\mathcal{T}}

\def\DD{\mathcal{D}}

\def\F{\mathcal{F}}
\def\G{\mathcal{G}}
\def\H{\mathcal{H}}

\def\gl{\mathfrak{gl}}

\def\g{\mathfrak{g}}
\def\h{\mathfrak{h}}

\def\a{\mathbf{a}}
\def\d{\mathbf{d}}
\def\q{\mathbf{q}}

\def\0{\mathbf{0}}

\newcommand{\DO}[2]{\mathrm{DO}(#1, #2)}
\newcommand{\PSIDO}[1]{\Psi\mathrm{DO}({#1})}
\newcommand{\IPSIDO}[1]{\mathrm{I}\Psi\mathrm{DO}({#1})}
\newcommand{\IDO}[2]{\mathrm{IDO}(#1, #2)}

\begin{document}
\begin{abstract} In this paper we describe a Poisson pencil associated to the lattice $W_m$-algebras defined in \cite{IM}, and we prove that the Poisson pencil is equal to the one defined in \cite{MW} and \cite{CM} using a type of discrete Drinfel'd-Sokolov reduction. We then show that, much as in the continuous case, a family of Hamiltonians defined by fractional powers of difference operators commute with respect to both structures, defining the kernel of one of them and creating an integrable hierarchy in the Liouville sense.
\end{abstract}
\maketitle
\section{Introduction and Background}

The space of operators of the form $u_0(x) + \ldots + u_{m-2}(x) \partial^{m-2} - \partial^m,
$
where $\partial := \partial / \partial x$ and the coefficients $u_i(x)$ are periodic functions, has a remarkable quadratic Poisson structure, called the second Adler-Gelfand-Dickey bracket \cite{GD, Adler}, and defined by Lax almost 50 years ago. Its Poisson algebra is known as the classical $W_m$-algebra. Adler-Gelfand-Dickey structures are perhaps better known  through its connection to integrable systems as they are Poisson brackets for KdV-type equations \cite{GD}. These equations are biHamiltonian, i.e. Hamiltonian with respect to two compatible Poisson structures whose sum is also Poisson. All best-known integrable systems are biHamiltonian. The Adler-Gelfand-Dickey bracket can be constructed in at least two equivalent ways: as defined by the multiplicative structure on the algebra of formal pseudo-differential operators  interpreted as a Poisson-Lie structure on the extended group of such operators \cite{khesin1995poisson}); or defined through a Drinfeld-Sokolov reduction on the dual of an affine (Kac-Moody) Lie algebra \cite{DS}.

In a recent paper \cite{IM}, the authors defined a discrete version of these constructions, the lattice $W_n$-algebra. This Poisson algebra can also be constructed in two different ways, similarly to the continuous case. One can define a multiplicative structure on the space of monic $m$ order difference operators of the form $(-1)^{m-1} +u^1\T+ \ldots + u^{m-2} \T^{m-1} - \T^m
$, where $u^i$ are bi-infinite $N$-periodic sequences, that is $u^i = (u^i_n)_{n=-\infty}^{+\infty}$ with $u^i_{n+N} = u^i_n$ for any $n$, and where $\T$ acts on the space of bi-infinite periodic sequences $u = (u_n)$ by shifting the subindex once, i.e. $(\T u)_n = u_{n+1}$. The space of any $m$ order difference operators has a natural multiplicative structure that allow us to define a natural Poisson bracket. The authors in \cite{IM} showed that the resulting bracket can be reduced using left and right multiplication by bi-infinite $N$-periodic sequences to a bracket on the space of monic operators with constant zero term, as above, to define what they called the $W_m$-lattice algebra. 

This bracket had an earlier definition through a modified discrete version of the Drinfel'd-Sokolov reduction. In \cite{MW} the authors proved that a bracket introduced by Semenov-Tian-Shansky in \cite{semenov85} could be reduced to a quotient of the form $\PSL(m+1)^N/H^N$ where $\RP^m = \PSL(m+1)/H$ is the homogeneous representation of the projective space, with $H^N$ acting on $\PSL(m+1)^N$ by discrete gauges. The resulting bracket coincides with the reduced bracket on monic operators, as shown in \cite{IM}. The authors also identified a second bracket defined through the same reduction process, but failed to prove that the two brackets were compatible. The authors of \cite{CM} naturally connected the Hamiltonians with respect to the quadratic bracket to invariant evolutions of projective polygons, lifting the two Poisson structures to pre-symplectic forms on the space of projectively invariant polygonal vector fields. They used this connection to show that the two brackets  where compatible and some associated evolutions were biHamiltonian. Through these general constructions one can recover familiar structures that have appeared in the literature as Hamiltonian structures for the lattice Visasoro algebra or Volterra lattice \cite{volkov1988miura, faddeev2016liouville}, and the lattice $W_3$-algebra \cite{belov1993lattice}. \par

In this paper we aim to describe this Hamiltonian pencil using the $W_m$-algebra definition as in \cite{IM}, and we will prove that the companion bracket to the $W_m$-algebra coincides with the companion bracket defined in \cite{MW}. This interpretation will allow us to readily identify a Liouville integrable system defined by Hamiltonians defined by the traces of fractional powers of the difference operators 
\[
\F^s (D) = \sum_n \tr(D^{s/m}),
\]
much like in the continuous case. The proof of the equivalence of both pencils is achieved through the identification of the pre-symplectic forms $\omega_i$, $i=1,2$, that lift both Poisson structures $\{,\}_i$, $i=1,2$, to the space of invariant vector fields on twisted polygons in centro-affine geometry, that is, the case of arbitrary $u^0$ (rather than constant). 

Finally, we will show that if the lift of an $\F$-Hamiltonian evolution with respect to $\{,\}_1$ to a polygonal vector field is denoted by $X^\F$, then $X^\F$ is the Hamiltonian vector field with respect to the pre-symplectic form $\omega_1$, for both centro-affine and projective cases. In particular, in the projective case $X^{\F^s}$ is defined by the nonnegative part of $D^{s/m}$, for any $s$.

The author is deeply grateful to Professor Anton Izosimov for discussions and for his input on the content of section 5. His suggestions and ideas facilitated the results presented in this paper.

\section{A discretization of the Adler-Gelfand-Dikii bracket: Discrete $W_m$-algebras}\label{sec:scalar} 

We denote the space of $N$-periodic upper-triangular difference operators of order $m$ by $\DO{N}{m}$. That is, its elements are of the form
\begin{equation}\label{app:operator}
D =\sum_{i = 0}^m a^i \T^i,
\end{equation}
where $a_i$'s are $N$-periodic functions $\Z \to \RR$ acting on functions of the same kind by term-wise multiplication, while $\T$ is the \textit{left} shift operator $(\T f)(x) = f(x+1) $ (the term \textit{upper-triangular} is used to distinguish such operators from those which may also contain terms of negative power in $\T$). We define an $N$-periodic \textit{pseudodifference operator} as an expression of the form
\begin{equation}\label{app:psido}
\sum_{i = -\infty}^k b^i \T^i,
\end{equation}
where $k\in \Z$, and each $b^i \colon  \Z \to \RR$ is an $N$-periodic function. Such an expression can be regarded either as a formal sum, or as an actual operator acting on the space $\{\xi \colon \Z \to \RR \mid \exists\, j \in \Z : \xi(x) = 0 \, \forall \, x > j\}$ of eventually vanishing functions.\par
We will denote the set of $N$-periodic pseudodifference operators by $\PSIDO{N}$. This set is an associative algebra. Moreover, almost every pseudodifference operator is invertible. In particular, \eqref{app:psido} is invertible if the coefficient $b^k$ of highest power in $T$ is a non-vanishing sequence. 
We will denote the set of invertible $N$-periodic pseudodifference operators by $\IPSIDO{N}$. This is a group with respect to multiplication. At least formally, one can regard it as an infinite-dimensional Lie group. 

The following proposition was proved in \cite{IM}
\begin{proposition}\label{app:mainprop}
There exists a natural Poisson structure $\pi$ on the group $\IPSIDO{N}$ of $N$-periodic invertible pseudodifference operators. This structure has the following properties:
\begin{enumerate}
\item[1.]  It is multiplicative, in the sense that the group multiplication is a Poisson map. In other words, the group $\IPSIDO{N}$, together with the structure $\pi$, is a \textit{Poisson-Lie group}.
\item[2.] The subset $\IDO{N}{k} := \IPSIDO{N} \cap \DO{N}{k}$ of order $k$ invertible upper-triangular difference operators is a Poisson submanifold.
\item[3.]The Poisson structure $\pi$ vanishes on the submanifold  $\IDO{N}{0}$ of invertible order zero operators.
\item[4.]\label{4} The Poisson structure $\pi$ is invariant under an automorphism of $\IPSIDO{N}$ given by conjugation $\mathcal D \to f\mathcal D f^{-1}$ with quasiperiodic $f \colon \Z \to \RR$.
\end{enumerate}
\end{proposition}

The natural Poisson structure defined above appears on any Lie group which is embedded as an open subset into an associative multiplicative algebra $A$ (for example, as its invertible elements).  In that case, the Lie algebra of $G$ (or the tangent space to $G$ at any point) can be naturally identified with $A$. Assume also that $A$ is endowed with an invariant inner product, that is, $(xy,z) = (x,yz)$ for any $x,y,z \in A$ (in particular, the inner product is adjoint invariant). Furthermore, assume that $r \colon A \to A$ is a skew-symmetric operator satisfying the modified Yang-Baxter equation
\begin{equation}\label{app:mybe}
[rx, r y] - r[rx, y] - r[x, r y] = -[x,  y] \quad \forall\, x,y \in \g.
\end{equation}
Then $G$ carries a structure of a factorizable Poisson-Lie group. Identifying the cotangent space $T_g^*G$ with the tangent space $T_g G = A$ by means of the invariant inner product, one can then write the formula for the corresponding Poisson tensor on $G$ as
\begin{equation}\label{app:gd}
\pi_g(x,y) = (r(xg), yg) - (r(gx), gy) \quad \forall\, g \in G, x,y \in A.
\end{equation}
Property 4 in proposition \ref{app:mainprop} allows us to reduce the natural Poisson bracket on $\IDO{N}{m}$ given by operators as in (\ref{app:operator}), to difference operators where $a^m = -1$ and $a^0 = (-1)^{m-1}$, both constant. The reasons for these particular choices will become clear in our next section.

With this construction in mind, consider the inner product on $\PSIDO{N}$ given by
\[
\langle V, W\rangle  = \sum_{n=0}^N \tr(VW)(n)
\]
where if $V = \sum v^r \T^r$, $\tr(V) = v^0$. This inner product is invariant and can be used to define a Poisson bracket on the space of invertible difference operators. Indeed, if
$\F :\IDO N m \to \RR$, its variational derivative is represented by a pseudo-difference operator of order $m$, denoted here by $\deltaF$ and defined uniquely by
\[
\frac d{d\epsilon}|_{\epsilon=0} \F(D(\epsilon)) = \langle \deltaF, \frac d{d\epsilon}|_{\epsilon=0} D(\epsilon)\rangle.
\]
With this notation, the quadratic Poisson bracket above becomes
\begin{equation}\label{firstB}
\{\F,\G\}(D)  = \sum_{n-1}^N \tr\left( r(D \deltaF) D- D r(\deltaF D), \deltaG\right)(n)
\end{equation}
where $r(L) = \frac12(L_+-L_-)$. Notice that we can substitute $r$ by $r^+(L) = L_++\frac 12 L_0$ and obtain the same bracket.

If the bracket is reduced, it will have an identical formula, with $\deltaF$ modified in standard fashion by the corresponding reduction. For more details about this construction, please see \cite{IM}.

\section{A discretization of the Drinfel'd-Sokolov reduction}  

The authors of \cite{MW} defined a pair of Poisson structures associated to the background geometry of projective twisted polygons in $\RP^{m-1}$. The pair was shown to be compatible in a subsequent paper \cite{CM}, where the authors linked them to pre-symplectic forms on the space of projectively-invariant polygonal vector fields. In this section we will briefly recount the definition of the pair, defined through a discrete version of the well-known Drinfel'd-Sokolov reduction \cite{DS}.

As a homogenous space, the projective space $\RP^{m-1}$ can be described as $\PSL(m)/H$, where the subgroup $H$ is the isotropy subgroup of a distinguished point. The projective group $\PSL(m)$ acts on the quotient via left multiplication on class representatives. We say an infinite polygon in $\RP^{m-1}$ is {\it $N$-twisted} is there exists an element of the projective group $M\in \PSL(m)$ called the {\it monodromy}, such that, if $\gamma_n$ is the $n$th vertex, then $\gamma_{n+N} = M\cdot \gamma_n$, for all $n$. We focus on the moduli space of equivalent classes of polygons in $\RP^{m-1}$ under the action of the projective group, and define coordinates for this space. The authors of \cite{MW} proved that an $N$-twisted projective polygon $\gamma$ completely determines the solution of a recursion equation of the form
\begin{equation}\label{invariants}
x_{n+m} = a_n^{m-1} x_{n+m-1} + a_n^{m-2} x_{n+m-2} + \dots+ a_n^1 x_{n+1} + (-1)^{m-1}x_n
\end{equation}
for any $n$, up to the action of the projective group on $\gamma$, whenever $N$ and $m$ are co-prime (the case $m=3$ appeared in \cite{OST}). The solution $x$ is defined by the entries of a unique lift of $\gamma$ to $\RR^m$, more about this in our next section. The discrete functions  $a^k$, $k=1,\dots,m-1$ define bi-infinite $N$-periodic sequences, which are invariant under the projective action of $\PSL(m)$ on the polygon $\gamma$. Therefore, they can be considered to be coordinates in the moduli space. The same description holds if we consider the centroaffine space $\RR^m$ with $\GL(m)$ acting on it linearly. In this case the invariant coordinate $a_n^0$, coefficient of $x_n$, will not be constant.

The following proposition was proven in \cite{MW} for the projective case, below is the $\GL(m)$ case in \cite{IM}. The proofs are identical and we do not include it.
\begin{proposition} The moduli space of non-degenerate twisted polygons under the linear action of $\GL(m)$ can be identified with an open and dense subset of $\GL^N(m)/H^N$, where $H\subset \GL(m)$ is the subgroup 
\(
H = \{ g\in G, ~~g e_1 =  e_1 \}.
\)
$H^N$ acts on $\GL(m)^N$ via the right discrete gauge action 
\begin{equation}\label{gauge}
(h, g) \to ((\T h) g h^{-1})
\end{equation}
with $h\in H^N$ and $g\in \GL(m)^N$ representing a bi-infinite $N$-periodic sequence.
\end{proposition}

Finally, let  
\[
r = \sum_{i>j} E_{ij}\otimes E_{ji} + \frac 12\sum_r E_{rr}\otimes E_{rr}
\]
be the standard $r$ matrix for $G=\GL(m)$, where $E_{ij}$ has a $1$ in the entry $(i,j)$ and zeroes elsewhere. Given $\F, \H$ smooth scalar-valued functions on $G^N$ and $A\in G^N,$ the \emph{twisted Poisson bracket} is defined as in \cite{FRS}:
\begin{equation}\label{twisted}
\begin{split}
\{\F, \H\}(A) & := \sum_{s=1}^Nr(\nabla_s \F \wedge \nabla_s \H) + \sum_{s=1}^N r(\nabla_s' \F \wedge \nabla_s' \H)  \\ &- \sum_{s=1}^N r\left( (\T\otimes 1)(\nabla_s' \F \otimes \nabla_s \H)\right) + \sum_{s=1}^N r  \left((\T\otimes 1) (\nabla_s' \H \otimes \nabla_s \F)\right), 
\end{split}
\end{equation}
where $\xi\wedge \eta=\tfrac12( \xi\otimes\eta - \eta\otimes\xi)$,  and $\nabla' \F$ and $\nabla \F$ are the right and left gradient, respectively.
Equation~\eqref{twisted} defines a Hamiltonian structure on $G^{N}$, as shown by Semenov-Tian-Shansky in~ \cite{semenov85}.  Moreover, the {\em right gauge action} of $G^{N}$ on itself 
\begin{equation}\label{gauge}
(g, A) \to ((\T g) A g^{-1})
\end{equation}
is a Poisson map and its orbits coincide with the symplectic leaves \cite{FRS, semenov85}.

\begin{theorem}\label{reducedbracket} (\cite{IM})
The Poisson bracket (\ref{twisted}) reduces locally to the quotient $G^N/H^N$, where $H^N$ is acting on $G^N$ via de right gauge action.
\end{theorem}

This theorem will naturally define a Poisson bracket on the open a dense subset of $G^N/H^N$ with $a^i$ defined in (\ref{invariants}) as coordinates. In our next section, we will describe how both brackets, the one defining the $W_m$ algebras, and the one defined by the discrete Drinfel'd-Sokolov reduction, coincide. The proofs and all details can be found in \cite{IM}. 
\section{Connection between both brackets} 

There is a natural connection between evolutions of difference operators of the form (\ref{app:operator}) with $a^m = -1$, and those of projective polygons defined by the kernel of the operators. This relation also exists in many other geometric backgrounds, including centro-affine geometry (the case when $G = \GL(m)$). In this section we remind the reader about this connection for both centro affine and projective cases, and we summarize the results in \cite{IM} that used this connection to prove that both brackets defined in previous sections coincide.

Assume $\gamma\in(\RR^m)^N$ defines a twisted polygon in $\RR^m$, twisted with respect to the linear action of $\GL(m)$ on $\RR^m$. Simply from dimensional reasons, there exist $a^k\in \RR^N$, $k=0,1,\dots,m-1$ such that 
\begin{equation}\label{kernel}
\T^m\gamma = a^{m-1}\T^{m-1}\gamma+a^{m-2}\T^{m-2}\gamma +\dots+a^0 \gamma.
\end{equation}
As proven in \cite{OST} for $m=3$ and in \cite{MW} for any $m$, given a $N$-twisted projective polygon, there exists a unique lift to $\RR^m$ satisfying (\ref{kernel}) with $a^0 = (-1)^{m-1}$, whenever $N$ and $m$ are co-primes. The following theorem was proven in \cite{MW} for the projective case and in \cite{IM} for the centro-affine case. The proof was constructive, showing that the vector field could be obtained explicitly and algebraically from the variation of the Hamiltonian.
\begin{theorem}
Let $a^k(t)$, $k=0,1,2\dots,m-1$ be the coordinates of a solution to an evolution that is Hamiltonian with respect to the reduced bracket defined in subsection \ref{reducedbracket}, with Hamiltonian $f({\bf a})$. Let $D(t)$ be defined by
\begin{equation}\label{Da}
D(t) = - \T^m+a^{m-1}(t) \T^{m-1}+\dots+a^1(t)\T + a^0(t) 
\end{equation}
and let $\gamma(t)$ be the twisted polygon in $\RR^m$ defined by its kernel, $D(\gamma) = 0$. There exists a unique polygonal vector field $X^f$ in an open and dense subset of $(\RR^m)^N$ such that 
\begin{equation}\label{gammaev}
\gamma_t = X^f(\gamma(t)).
\end{equation}
And vice-versa, if $\gamma$ is a solution of (\ref{gammaev}), then the invariants $a^k$ will evolve following an $f$-Hamiltonian evolution with respect to the bracket in theorem \ref{reducedbracket}.
\end{theorem}

It is worth to briefly describe $X^f$'s connection to the invariants $a^k$, as we will use it in our next section. Let $\rho = (\gamma, \T \gamma, \dots, \T^{m-1}\gamma)$ and assume $\gamma_t = X^f$, where $f(\a)$ is an invariant Hamiltonian function. Then 
\begin{equation}\label{Q}
\rho_t = \rho Q^{X^f}
\end{equation} 
for some invariant matrix $Q^{X^f}$ depending on $a^k$ and their shifts. The matrix $Q^{X^f}$ defines $X^f$ and its shifts and it is directly connected to the group right and left gradients appearing in (\ref{twisted}), as we will see  in our next section where it will be widely used.

Moving now to the $W_m$-algebra picture, one can readily find the $\gamma$ evolutions that are directly linked to evolutions of difference operators which are Hamiltonian with respect to the bracket (\ref{firstB}). If $D(\gamma) = 0$, then $D_t (\gamma) + D(\gamma_t) = 0$ and if $D$ is Hamiltonian with respect to (\ref{firstB}), with Hamiltonian $\F$, then 
\[
D_t = r(D \deltaF) D- D r(\deltaF D)
\]
and so $D(\gamma_t) = -D_t(\gamma) = D r(\deltaF D)(\gamma)$. We call $Y^\F$ the vector field
\[
Y^\F = r(\deltaF D)(\gamma).
\]
\begin{theorem} (\cite{IM}) If $\F(D) = f({\bf a})$ whenever $D$ and $a$ are related as in (\ref{Da}), then
\[
X^f = Y^\F
\]
along $\gamma$. As a corollary, both brackets (\ref{firstB}) and the one in theorem \ref{reducedbracket} coincide when defined on the coordinates ${\bf a}= (a^k)$.
\end{theorem}

\section{A Poisson pencil}
  The space $\DO{N}{m}$ is a Poisson submanifold with the quadratic bracket (\ref{firstB}) defined on the space or difference operators. In this section we will identify its compatible bracket. Consider the $1$-parameter family of maps
\[
\phi_\lambda : \DO{N}{m} \to \DO{N}{m}
\]
with $\lambda \in \RR$, defined as $\phi_\lambda (\sum_{i = 0}^m a^i \T^i) = \sum_{i = 1}^m a^i \T^i + \lambda^{-1} a^0$. Furthermore, consider the push-forward of the quadratic bracket (\ref{firstB})
\(
\{, \}_\lambda = (\phi_\lambda)_\ast (\{,\})\) defined as
\begin{equation}\label{pencil}
\{\F, \G\}_\lambda({\mathcal D}) = \{\F\circ\phi_\lambda, \G\circ \phi_\lambda\}(\phi_\lambda^{-1}({\mathcal D})).  
\end{equation}
Clearly $\{,\}_\lambda$ is Poisson for any $\lambda \ne 0$. 

\begin{theorem} The bracket (\ref{pencil}) is a Poisson bracket for any $\lambda\in \RR$. In fact, it is a Poisson pencil, that is, linear in $\lambda$.
\end{theorem}
\begin{proof}
The proof is a straightforward calculation. For simplicity we will denote the operator $\delta_D\F$ by $V$ and $\delta_D\G$ by $W$. We will denote by $V_{+}$ (resp. $V_{-}$)  its positive (resp. negative) part as operator, and $V_0$ its zero term. Therefore
\[
\phi_\lambda^{-1}(D) = \lambda D_0 + D
\]
and 
\[
\delta (\F\circ\phi_\lambda )= \phi_\lambda^\ast V = \lambda^{-1} V_0 + V_-.
\]
(Notice that $V_+ = 0$). We will next substitute these values in (\ref{pencil}) and use (\ref{firstB}). This will require calculating individual terms, which we will do next.
\[
r^+(\delta (\F\circ\phi_\lambda \phi_\lambda^{-1}(D)) = r^+((\lambda^{-1} V_0 + V_-)(\lambda D_0 + \DD)) = \lambda^{-1}V_0D_++\frac 12 V_0 D_0+r^+(V_-D_+),
\]
and from here
\begin{multline}\label{term1}
\phi^\ast_\lambda \sum_n\tr\left( D_n r^+(\deltaF D_n), \deltaG\right) 
\\ = \sum_n\tr\left(\lambda^{-1}V_0D_++\frac 12 V_0 D_0+r^+(V_-D_+), , (W_-+\lambda^{-1}W_0)(D_+\lambda D_0)\right).
\end{multline}
The terms in 
\begin{equation}\label{term2}
\phi^\ast_\lambda \sum_n\tr\left( r^+(D_n\deltaF )D_n, \deltaG\right)
\end{equation}
will be analogous but with the factors in different order. When both terms are brought together in (\ref{pencil}), the terms in (\ref{term1}) will carry a negative sign.

The coefficient of $\lambda^{-2}$ in (\ref{term1}) is given by $\tr(V_0D_+W_0D_+)$, and the corresponding coefficient in (\ref{term2}) can be equally calculated to be $\tr(D_+V_0D_+W_0) = \tr(V_0D_+W_0D_+)$. Therefore, they will cancel in (\ref{pencil}).

The coefficient of $\lambda^{-1}$ in (\ref{term1}) is
\begin{multline*}
\tr(V_0 D_+W_-D_++V_0 D_+W_0D_0+\frac12 V_0 D_0W_0D_++ r^+(V_- D_+)W_0D_+) \\= \tr(V_0 D_+W_-D_+)\hskip2in
\end{multline*}
and the corresponding one for (\ref{term2}) can be equally calculated to be $\tr(D_+V_0 D_+W_-)$. Thus, the $\lambda^{-1}$ term in (\ref{pencil}) equals
\[
\sum_n -\tr(V_0 D_+W_-D_+) + \tr(D_+V_0 D_+W_-) = 0
\]
Therefore, the expression in (\ref{pencil}) is indeed a pencil and the coefficients of $\lambda^1$ and $\lambda^0$ define compatible Poisson brackets. 
\end{proof}
One can find a companion bracket straightforwardly. The coefficient of $\lambda$ in (\ref{term1}) is given by
\[
\tr\left((\frac 12 V_0D_0+r^+(V_-D_+))W_-D_0\right) = \tr(r^+(V_-D_+))W_-D_0).
\]
And when placed in (\ref{pencil}) together with its counterpart in (\ref{term2}), we have the $\lambda$ coefficient of (\ref{pencil}) to be
\[
\{\F, \G\}_2(D) = \sum_n \tr\left( r^+(D_+V_-)D_0 W_-- r^+(V_-D_+)W_-D_0 \right).
\]
This is the companion bracket to our original quadratic bracket. While this bracket is also quadratic, upon reduction to $\SL(m)$ we do have a linear bracket since $D_0 = (-1)^{m-1}$. In that case, the bracket becomes
\begin{multline}\label{secondB}
\{\F, \G\}_2(D) = (-1)^{m-1}\sum_n \tr\left([(D_+V_-)_+ + \frac 12 (D_+V_-)_0] W_- - [(V_-D_+)_+W_--\frac12(V_-D_+)_0 W_-\right) 
\\ = (-1)^{m-1}\sum_n\tr([D_+, V_-]_+ W).
\end{multline}

\section{Pre-symplectic forms on polygonal vector fields and the equivalence of Poisson pencils}
The authors of \cite{CM} defined a pair of Poisson brackets and lifted them to pre-symplectic forms on projectively invariant vector fields on polygons in $\RP^m$. In this section we will describe the corresponding pre-symplectic forms for the centro-affine case ($a^0$ generic), and we will show that the Poisson pencil just found on our previous section is equal to the one found in \cite{CM} for the projective case. Consider the two Poisson tensors generating the pencil
\begin{equation}\label{firstB}
\{\F,\G\}_1(D)  = \sum_n\tr\left( r(D_n \deltaF) D_n- D_n r(\deltaF D_n), \deltaG\right)
\end{equation}
where $r(L) = \frac12(L_+-L_-)$, and
\begin{equation}\label{linear}
\{\F,\G\}_2(D) = \sum_n \tr\left(\left[D_0((\deltaF)_-(D_n)_+)_+ - ((D_n)_+(\deltaF)_-)_+D_0\right]\deltaG\right)
\end{equation}
Asume $X^f$ is defined as in (\ref{gammaev}), and define $\rho_n =(\gamma_n, \gamma_{n+1},\dots,\gamma_{n+m-1})$ and $\d_n = \det\rho_n$. In the projective case $\d_n = 1$ for all $n$ and $\gamma$ is a lift for the projective polygon determined uniquely by that property (\cite{MW}). Define also the discrete form $\theta = (\theta_n)$ along polygons given by 
\[
\theta_n(X) = \det(X_n,\gamma_{n+1},\dots,\gamma_{n+m-1})
\]
whenever $X$ is a vector field along $\gamma$ in $\RR^m$.

Let $\F: DO(N,m) \to \RR$ and let $f:\RR^{(m+1)N}\to \RR$ be defined as
\[
\F(\left(\sum^{m}_{r=0} a^r\T^r\right)) = f((a^r)).
\]
From the definition of variational derivative, we can see that
\[
\deltaF = \sum_{r=0}^m \T^{-r} \delta_{a^r} f
\]
where $\delta_{a^r} f$ is the standard variational derivative of $f$ in the $a^r$ direction. Recall that the reduction of both brackets to the $\SL(m)$ case is explicitly achieved using the left and right multiplication so that if 
\[
L = \sum_{k=0}^m a^k \T^k
\]
we reduce by finding $a, b$ such that $a^{-1} b^{-1} L b = D$. The same process is used if we are to work in the $\GL(m)$ case, in which case we only need to find $b$ so that $a^m = -1$, while $a^0$ is still unconstrained.

If we equate the $m$ terms,  $  b^{-1} a^m b_m = -1$, or $b^{-1} b_m = -\frac{1}{a^m}$, which has a unique solution if $(N, m) = 1$. Notice that if  $a^m = -1$, then $b = 1$. In any case, the reduced variational derivative in $\GL(m)$ ($a^m = -1$) will look like
\[
\deltaF = \sum_{r=0}^{m-1} \left[\T^{-r} \delta_{a^r} f + \T^{-m} \beta^r \delta_{a^r} f \right]
\]
for some $\beta^r$ that can easily be found explicitly, but which will cancel in our calculations and hence we do not need to know.

\begin{theorem}  Assume $N$ and $m$ are co-prime. Then, the Poisson bracket (\ref{linear}) satisfies
\[
\{\F,\G\}_2(D_n) = (-1)^{m-1}\omega_2(X^g, X^f)
\]
where $\omega_2 = \sum_nd (\frac1{\d_n}\theta_n)$, that is
\begin{eqnarray*}
\omega_2 (X,Y) &=& \sum_n \frac 1{\d_{n}}\left[Y\theta_n(X)-X\theta_n(Y)-\theta_n([Y,X])+\frac 1{\d_n}(\theta_n(Y) X(\d_n) - \theta_n(X) Y(\d_n)) \right],
\end{eqnarray*}
and where $X^f$ is as in (\ref{gammaev}).
\end{theorem}
\begin{proof} 

Straightforward calculations show that 
\begin{eqnarray*}
&&D_0((\deltaF)_-(D_n)_+)_+ - ((D_n)_+(\deltaF)_-)_+D_0\\&=&\sum_{r=1}^{m-1}\left[\T^{m-r}a^0\delta_{a_n^r}f-a^0\T^{-r}\delta_{a_n^r}f\T^m\right]
+\sum_{r=1}^{m-2}\sum_{s=r+1}^{m-1}\left[a^0\T^{-r}a_n^s\delta_{a_n^r}f\T^s- a_n^s\T^{s-r}\delta_{a_n^r}f a^0\right].
\end{eqnarray*} 
Using this expression, we get 
\begin{eqnarray}\label{first}
\{\F,\G\}_2(D_n) &=& \sum_n \sum_{r=1}^{m-1} \delta_{a_n^{m-r}}g \left[\T^{m-r}a^0\delta_{a_n^r}f -a^0\T^{-r}\delta_{a_n^r}f \T^m\right]\T^{r-m}\\&+& \sum_{r=1}^{m-2}\sum_{s=r+1}^{m-1}\delta_{a_n^{s-r}}g\left[a^0\T^{-r} a_n^s\delta_{a_n^r}f\T^s-a_n^s\T^{s-r}a^0\delta_{a_n^r}f\right]\T^{r-s}.
\end{eqnarray}

Let $Q_n^X\in \gl(m)$ be  defined as in (\ref{Q}), that is, defined by the relation 
\[
X(\rho) = \rho Q^X.
\]
If $X = X^f$ is as in (\ref{gammaev}), and $A$ is defined by $\rho_{n+1} = \rho_n A$, then $X(A) = A \T Q - Q A$, and

\[
Q^X = (\q, A\T\q,\dots,A\T A\T A\dots \T A\T \q)
\]
where $\q$ is an invariant vector defined by $X(\gamma) = \rho\q$ and $\T$ appears $m-1$ times in the last column of the matrix above. As pointed out before, the matrix $Q^X$ was explicitly related to the left and right gradients in the Lie group (called $\nabla \F$ and $\nabla' \F$ in (\ref{twisted})) in \cite{IM} and we will refer the reader to that paper for details of this relation, we will simply quote them next. 

We know from \cite{IM} that if $X = X^f$, then
\begin{equation}\label{first row}
Q_n^Xe_1 = \frac 12(\nabla_{n-1}\F + \nabla_n'\F)e_1
\end{equation}
(equation (35) in \cite{IM}), where 

\begin{equation}\label{nablas}
\nabla'_nF = \begin{pmatrix} -a^0_n\delta_{a_n^0} f& -a^0_n(\delta_{{\bf a}_n} f)^T\\ \ast&\ast\end{pmatrix}, \quad \nabla_n F = \begin{pmatrix} \ast & \ast\\ -(\delta_{{\bf a}_{n}} f)^T& -a_n^0\delta_{a^0_{n}} f-{\bf a}_{n}\cdot \delta_{{\bf a}_{n}} f\end{pmatrix}
\end{equation}
(Lemma 4.5 in \cite{IM}). Recall that since (\ref{twisted}) has been reduced using the Lie group discrete gauge action of $H^N$, we have that
\[
 \nabla'_{n+1} F - \nabla_n F \in \h^0
 \]
  for any $n$, where $\h$ is the Lie algebra of $H$, and $\h^0$ is its annihilator.

Using this information we can find both $Q^X_{1,r+1}$ and $Q^X_{r+1,1}$. Indeed
\[
Q^Xe_{r+1} = K\T Qe_r \T^{-1}= K\T K \dots \T K \T Q e_1\T^{-r} = \frac 12 K\T K \dots \T K (\nabla \F + \T\nabla'\F\T^{-1}) e_1\T^{-r}
\]
with $\T$ appearing $r$ times. We notice that when calculating $Q^X_{1,r+1}$ no entries from the first row of  $\nabla'\F$ are involved. Therefore, since 
\begin{equation}\label{h0}
\T^{-1}\nabla_n\F\T - \nabla_n'\F \in \h^0
\end{equation} 
we can substitute $\nabla\F$ by $\T\nabla'\F\T^{-1}$ above to obtain
\[
e_1^TQ^Xe_{r+1} = K\T K \dots \T K (\nabla\F) e_1\T^{-r}.
\]
Finally, $K \nabla\F = \nabla' \F K$, and so
\begin{eqnarray*}
e_1^TQ^Xe_{r+1} &=&e_1^T K\T K \dots \T K \T(\nabla'\F) K e_1\T^{-r} = e_1^TK\T K \dots \T K\T (\nabla'\F) e_2\T^{-r}\\&=&e_1^TK\T K\dots\T K \nabla\F e_2\T^{-r+1}= e_1^TK\T K\dots\T\nabla'\F K e_2\T^{-r+1} = \dots= e_1^T K\T\nabla'\F e_r\T^{-1}
\end{eqnarray*}
\begin{equation}\label{Q1r+1}
Q_{1,r+1} = e_1^T\nabla'\F e_{r+1} = -a^0 \delta_{a^{r}} f
\end{equation}
for any $r=1,2,\dots,m-1$.

Also from (\ref{first row}), (\ref{h0}) and (\ref{nablas}) we see that if $r=2,\dots,m$
\begin{eqnarray*}
Q^X_{r,1} &=& \frac{e_r^T}2(\T^{-1}\nabla \F\T+\nabla'\F)e_1 = e_r^T\T^{-1}\nabla\F e_1\T= e_{r+1}^T\T^{-1}K\nabla\F e_1\T-\T^{-1}a^re_m^T\nabla\F e_1\T\\&=&e_{r+1}^T\T^{-1}K\nabla\F e_1\T+\T^{-1}a^r\delta_{a^1}f \T
=e_{r+1}^T\T^{-1}\nabla'\F K e_1\T+\T^{-1}a^r\delta_{a^1}f\T \\&=& e_{r+1}^T\T^{-1}\nabla'\F e_2\T+\T^{-1}a^r\delta_{a^1}f\T =e_{r+1}^T\T^{-2}\nabla\F e_2\T^2+\T^{-1}a^r\delta_{a^1}f\T
\end{eqnarray*}
Iterating this process we obtain that 
\begin{eqnarray}\label{Qr1}
Q_{r,1} &=& \sum_{s=1}^{m-r}\T^{-s}a^{s+r-1}\delta_{a^s} f \T^s+ e_m^T\T^{r-m-1}\nabla\F e_{m-r+1} \T^{m-r+1}\\ &=& \sum_{s=1}^{m-r}\T^{-s}a^{s+r-1}\delta_{a^s} f \T^{-s}-\T^{r-m-1}\delta_{a^{m-r+1}}f\T^{m-r+1}.
\end{eqnarray}

We are now ready to put everything together. Using all the information above we can conclude that
\begin{eqnarray*}
&&\sum_n \frac 1{\d_{n}}d(\theta_n)(Y,X) = \sum_n \frac 1{\d_{n}}\left[Y\theta_n(X)-X\theta_n(Y)-\theta_n([Y,X])\right]\\
&=&\sum_n\frac 1{\d_{n}}\sum_{k=1}^{m-1}[\det(X_n, \gamma_{n+1},\dots, Y_{n+k},\dots,\gamma_{n+m-1})- det(Y_n,\gamma_{n+1},\dots, X_{n+k},\dots,\gamma_{n+m-1})]\\
& =& \sum_n  \left[(Q_n^X)_{1,1}\tr(Q_n^Y)- (Q_n^Y)_{1,1}\tr(Q_n^X)+\sum_{k=1}^m\left[(Q_n^Y)_{k,1}(Q_n^X)_{1,k} - (Q_n^X)_{k,1}(Q_n^Y)_{1,k}\right]\right]\\
&=& \sum_n  \left[\frac1{\d^2_n}\theta_n(X)Y(\d_n)- \frac1{\d^2_n}\theta_n(Y)X(\d_n)+\sum_{k=1}^m\left[(Q_n^Y)_{k,1}(Q_n^X)_{1,k} - (Q_n^X)_{k,1}(Q_n^Y)_{1,k}\right]\right],
\end{eqnarray*}
 and substituting (\ref{Q1r+1}) and (\ref{Qr1}) above, we get, after some minor rewriting,
\begin{eqnarray*}
\omega_2(X^f, X^g) &=& \sum_n \left(\sum_{r=1}^{m-1}a_n^0 \delta_{a_{n-r}^r}g\delta_{a_{n}^{m-r}} f - \sum_{r=1}^{m-1}a_{n+m-r}^0\delta_{a_{n}^{m-r}} f\delta_{a_{n+m-r}^r}g\right) \\&-&\sum_n\left(\sum_{r=1}^{m-2}\sum_{s=1}^{m-r-1}a_n^0a_{n+s}^{r+s}\delta_{a_n^r}f\delta_{a_{n+s}^{s}} g-\sum_{r=1}^{m-2}\sum_{s=r+1}^{m-1}a_{n+m-s}^0a_n^{r+m-s}\delta_{a_{n+m-s}^r} g\delta_{a_n^{m-s}}f\right).
\end{eqnarray*}
We can compare this expression times $(-1)^{m-1}$ to (\ref{first}) and, after some straightforward modifications, conclude that they are equal. \end{proof}

\begin{theorem}  The Poisson bracket (\ref{firstB}) coincides with the negative of the bracket  previously defined as a reduction of (\ref{twisted}). Furthermore, let $D_n = \sum_{k=0}^{m-1} a_n^k \T^k - \T^{m}$ so that $D_n(\gamma) = 0$ for any lift of its associated projective polygon $\gamma$. Then
\[
\{\F,\G\}_1(D_n) = (-1)^{m-1}\omega_1(X^g, X^f)
\]
where 
\[
\omega_1 (X,Y) = \frac 12\sum_n \frac 1{\d_{n+1}}\left[ X\theta_n(D(Y)) - Y \theta_n(D(X)) - \theta_n(XD(Y)-YD(X)))\right]
\]
\[
- \frac1{\d_n\d_{n+1}}\left[\theta_n(D(Y)) X(\d_n)- \theta_n(D(X))Y(\d_n)\right].
\]

\end{theorem}
\begin{proof} The fact that the reduction of (\ref{twisted}) and (\ref{firstB}) are equal was proved in \cite{IM}. 

Given that $D_n(\gamma) = 0$ for all $n$, and differentiating in the direction of a vector field $Y$, we obtain
\[
(D_n)_{t_Y}(\gamma) + D_n Y = 0 \quad \to \quad D_n(Y) = -\sum_{s=0}^{m-1} (a_n^s)_{t_Y} \gamma_{n+s}.
\]
From here, if the $n+k$ field is placed in the $\gamma_{n+k}$ position in $\d_n$, we have 
\begin{eqnarray*}
&&X\theta_n(D(Y)) - Y \theta_n(DX)) - \theta_n(XD(Y)-YD(X))\\ &=& \sum_{k=1}^{m-1}\left[\det(D_n Y, \gamma_{n+1},\dots, X_{n+k},\dots,\gamma_{n+m-1})- \det(D_n X, \gamma_{n+1},\dots, Y_{n+k},\dots,\gamma_{n+m-1})\right]\\
&=& \sum_{k=1}^{m-1}(a_n^k)_{t_Y}\det(X_{n+k}, \gamma_{n+1},\dots, \gamma_{n+k},\dots,\gamma_{n+m-1})- \sum_{k=1}^{m-1}(a_n^k)_{t_X}\det(Y_{n+k}, \gamma_{n+1},\dots, \gamma_{n+k},\dots,\gamma_{n+m-1})\\
&-& (a_n^0)_{t_Y}\sum_{k=1}^{m-1}\det(\gamma_n, \gamma_{n+1},\dots, X_{n+k},\dots,\gamma_{n+m-1}) + (a_n^0)_{t_X}\sum_{k=1}^{m-1}\det(\gamma_n, \gamma_{n+1},\dots, Y_{n+k},\dots,\gamma_{n+m-1})\\
&=& \sum_{k=0}^{m-1}(a_n^k)_{t_Y}\det(X_{n+k}, \gamma_{n+1},\dots, \gamma_{n+k},\dots,\gamma_{n+m-1})- \sum_{k=0}^{m-1}(a_n^k)_{t_X}\det(Y_{n+k}, \gamma_{n+1},\dots, \gamma_{n+k},\dots,\gamma_{n+m-1})\\
&-& (a_n^0)_{t_Y} X(\d_n)+ (a_n^0)_{t_X}Y(\d_n)
.\end{eqnarray*}
As in the previous proof, and using (\ref{Q1r+1})
\[ 
\det(X^f_{n+k}, \gamma_{n+1},\dots, \gamma_{n+k},\dots,\gamma_{n+m-1}) = \d_n (Q_n^{X^f})_{1,k+1} = -\d_na_n^0 \delta_{a_n^k} f.
\]
Also, given that $D_n (\gamma) = 0$, we conclude that $a_n^0 = (-1)^{m-1}\frac{\d_{n+1}}{\d_n}$. Therefore
\[
\frac 1{\d_{n+1}}det(X^f_{n+k}, \gamma_{n+1},\dots, \gamma_{n+k},\dots,\gamma_{n+m-1}) = (-1)^{m} \delta_{a_n^k} f
\]
and
\begin{eqnarray*}
&&\sum_n\frac 1{\d_{n+1}} \sum_{k=0}^{m-1} (a_n^k)_{t_{X^g}}  \det(X^f_{n+k}, \gamma_{n+1},\dots, \gamma_{n+k},\dots,
\gamma_{n+m-1}) \\&=& (-1)^m\sum_n (a_n^k)_{t_{X^g}}\delta_{a_n^k} f = (-1)^{m-1}\{f, g\}_1(\a).
\end{eqnarray*}

Finally
\[
(a_n^0)_{t _X}= \frac 1{\d_n}\det(L_n(X), \gamma_{n+1},\dots,\gamma_{n+m-1})  = \frac1{\d_n}\theta_n(D(X)).
\]
The theorem follows.
\end{proof}
\begin{corollary} Both $\omega_1$ and $\omega_2$ are closed forms when defined on the space of invariant vector fields, and hence pre-symplectic. Furthermore, $X^f = Y^f$ is the $f$-Hamiltonian vector field with respect to $\omega_1$.
\end{corollary}
\begin{proof}
Using that $X^f$ induces the $f$-Hamiltonian evolution with respect to $\{,\}_1$ on the invariants $\a$, we get that $X^f$ is the Hamiltonian vector field for $\omega_1$ when restricted to invariant fields.  That is, $\omega_1(X^f, Y) = Y(f)$ for every $f$ invariant. We then get that $[X^f, X^g] = X^{\{f,g\}_1}$ and from here $\omega_1$ will be a closed form on the space of invariant vector fields since the property is equivalent to $\{,\}_1$ being Poisson. 

We obtain directly that $\omega_2$ is also closed since it is an exact form. 

\end{proof}

\section{Commuting family of Hamiltonians} 
In this section we will assume that we are working on the $\SL(m)$ case. Notice that while the previous results were proven for the $\GL(m)$ case, the $\SL(m)$ case is directly obtained through reduction. We will prove that,  as it happened in the continuous case, there exists a hierarchy of completely integrable systems with $\{,\}_1$-Hamiltonians given by

\begin{equation}\label{Fr}
\F_s (D) = \sum_n \tr(D_n^{s/m})
\end{equation}
for any $s=1,\dots,m-1$, where the fractional powers are naturally defined. We will do so by showing that the Hamiltonians above  are in involution with respect to both Poisson brackets. In addition to this integrable system there is an additional system, the Boussinesq Lattice, which is also biHamiltonian with respect to our pencil (see \cite{MW}), with $\{,\}_2$-Hamiltonian given by 
\[
\H(D) = \sum_n \ln a^1_n.
\]
The authors of \cite{MW} showed that $\H$ is in the kernel of the Poisson bracket $\{,\}_1-\{,\}_2$.
\begin{proposition}
Let $\F_s$ be defined as in (\ref{Fr}). Then the variational derivative of $\F_s$ after the reduction to $\SL(m)$ is given by the difference operator
\begin{equation}\label{Zs}
Z^s = \frac  sm D^{\frac{s-m}m} + (-1)^m\frac{s}m\tr D^{s/m}
\end{equation}
for any $s$. 
\end{proposition}
\begin{proof} First of all we will investigate the effect of the reduction by left and right multiplication on $L$ on the fractional power. Denote by $\hat L(\epsilon)  = a^{-1} b^{-1} (L+\epsilon V) b$ the reduced operator, and define $a' = \frac d{d\epsilon}|_{\epsilon=0} a$ with $a^0 = (-1)^{m-1}, a^m = -1$. Since $a^0(\epsilon) = \tr(L+\epsilon V)$ and $a^m = \tr(\T^{-m} (L+\epsilon V))$, we have that
\[
a' = (-1)^{m-1} (a^0)' = (-1)^{m-1}\tr V.
\]
 Define 
\[
Z(1,V) = ((\hat L(\epsilon))^{1/m})'|_D.
\]
Using $(\hat L^{1/m})^m = \hat L = a^{-1} b^{-1} (L+\epsilon V) b$, we conclude that, if we differentiate and evaluate at $D$, we obtain
\[
\sum_{s=0}^{m-1} D^{s/m} Z(1,V) D^{\frac{m-1-s}m}=V -a' D-b' L + L b'
\]
and from here
\[
\sum_{i=0}^{m-1} D^{i/m} Z(1,V) D^{-i/m}=(V -a' D-b' D + D b') D^{\frac{1-m}m}.
\]
Applying the trace 
\[
m\tr(Z(1,V)) = \tr(VD^{\frac{1-m}m}) +(-1)^m \tr V \tr D^{1/m} - \tr( b' D^{1/m} - D b' D^{\frac{1-m}m}) 
\]
\[= \tr( V\left(D^{\frac{1-m}m} + (-1)^m \tr D^{1/m}\right))
\]
and therefore $Z^1 = \frac  1 m \left(D^{\frac{1-m}m}+ (-1)^m \tr D^{1/m}\right)$. 

Likewise, if \[Z(s,V) = ((\hat L(\epsilon))^{s/m})'|_D\] we have that
\[
Z(s,V) =\sum_{i=0}^{s-1} D^{i/m} Z(1,V) D^{\frac{s-i-1}m} = \sum_{i=0}^{s-1} D^{i/m} \left[Z(1,V) D^{\frac{s-1}m}\right]D^{-i/m}.
\]
We also know that 
\[
\sum_{i=0}^{m-1} D^{i/m} Z(1,V)D^{\frac{s-1}m} D^{-i/m}\]
\[
=(V -a' D-b' D + D b') D^{\frac{1-m}m}D^{\frac{s-1}m} = (V -a' D-b' D + D b') D^{\frac{s-m}m}
\]
and from here
\[
m \tr(Z(1,V)D^{\frac{s-1}m}) = \tr (VD^{\frac{s-m}m} - a'D^{s/m})
\]
Putting everything together
\[
\tr(Z(s,V)) = s \tr(Z(1,V) D^{\frac{s-1}m}) = \frac sm \tr(V D^{\frac{s-m}m}- (-1)^{m-1}\tr V D^{s/m})  
\]
\[
= \frac sm \tr(V(D^{\frac{s-m}m} + (-1)^m \tr D^{s/m})
\]
and so $Z^s = \frac sm(D^{\frac{s-m}m}+ (-1)^{m}\tr D^{s/m})$, as stated.
\end{proof}
\begin{theorem}
The family $\{\F_s\}_{s=1}^\infty$ commute with respect to both (\ref{firstB}) and (\ref{linear}), and so they generate an integrable system hierarchy. Furthermore, the kernel of $\omega_2$ is at least $m-1$ dimensional.
\end{theorem}
\begin{proof}
The proof that they commute with respect to (\ref{linear}) is straightforward substituting $Z^s = \deltaF$ and $Z^p = \deltaG$  in (\ref{linear}) and observing that they both vanish when the $\T^0$ term in $D_n$ is constant and independent from $n$. Indeed $[Z^s_-, D_+]_+ = [Z^s_-, D]_+ = [D^\frac{s-m}m, D]_+ = 0$, and so $\F_s$ is in the kernel of (\ref{linear}) for all $s$, proving that the dimension of the kernel of $\omega_2$ is at least $m-1$. 

We also have that 
\[
\{\F_p,\F_s\}_1(D)  =\frac{p}{m} \langle r(D Z^p) D- D r(Z^pD), Z^s\rangle \]
\[
 = \frac{p}{2m}\langle D^{p/m}_+D+(-1)^m D_+\tr(D^{p/m})D- D^{p/m}_-D-DD_+^{p/m}+(-1)^{m-1}D\tr(D^{p/m}) D_++D D^{p/m}_-, Z^s\rangle
\]
\[
=\frac{p}{2m}\langle [D_+^{p/m}, D]+[\tr(D^{p/m}, D] + [D, D_-^{p/m}], Z^s\rangle = \frac{p}{m}\langle [D_+^{p/m}, D]+[\tr(D^{p/m}), D], Z^s\rangle
\]
where we have used that $[D^{p/m}, D] = 0$. 

Substituting $Z^s$ and noticing that the zero order term in $[D_+^{p/m}, D]+[\tr(D^{p/m}), D]$ vanishes when $a^0$ is constant, we get
\[
\{\F_p,\F_s\}_1(D) = \frac{ps}{m^2} \langle  [D_+^{p/m}, D]+[\tr(D^{p/m}), D], D^{\frac{s-m}m}\rangle. 
\]
We also have 
\[
 \tr([\tr(D^{p/m}), D] D^{\frac{s-m}m}) = \tr(\tr(D^{p/m} )[D, D^{\frac{s-m}m}])  = 0,
 \]
 \[
\tr([D_+^{p/m}, D]D^{\frac{s-m}m}) = \tr(D_+^{p/m}[D, D^{\frac{s-m}m}]) = 0,
 \]
and therefore, $\{\F_p, \F_s\}(D) = 0$ for any $p,s$.

\end{proof}
The existence of this hierarchy was conjectured in \cite{CM}  as linked to the two elements in the kernel of $\omega_2$ described in that paper (the statement in that paper is not correct, as the dimension is higher than $2$. A brief correction is forthcoming). Indeed, the two vector fields in that paper, $X^1$ and $X^2$ coincide with $X^{\F_1}$ and $X^{-\F_2}$, as shown next.

\begin{theorem} Let $X^f$ be the polygonal invariant vector field inducing the $f$-Hamiltonian on $D$, where $D(\gamma) =0$. Then $X^{\F_s}$ is defined by the nonnegative part of $D^{s/m}$
\[
X^{\F_s} = \frac sm(D^{s/m} - D_-^{s/m})(\gamma) = \frac sm(D^{s/m}_+ + \tr(D^{s/m}))(\gamma).
\]
\end{theorem}
\begin{proof}
Give the Hamiltonian $\F_s$ as in (\ref{Fr}) and its variation (\ref{Zs}), its Hamiltonian evolution with respect to (\ref{firstB}) is given by
\[
D_t = \frac 12\left([(DZ^s)_+-(D Z^s)_-]D
 -D[(Z^sD)_+-(Z^sD)_-] \right).
 \]
 Upon substituting (\ref{Zs}) we obtain
 \[
 D_t =  \frac s{2m}\left(D^{s/m}_+D - DD^{s/m}_++D D^{s/m}_-- D^{s/m}_-D +(-1)^m(D_+D_0^{s/m} D- D D_0^{s/m}D_+)\right).
 \]
 Next we observe that $[D, D^{s/m}] = 0$ and so $[D, D^{s/m}_-] = -[D, D^{s/m}_++D^{s/m}_0]$. We also note that
 \begin{eqnarray*}
  D_+D_0^{s/m} D- D D_0^{s/m}D_+&=& DD_0^{s/m} D- D D_0^{s/m}D - D_0D_0^{s/m} D+ D D_0^{s/m}D_0
  \\
  &=& (-1)^{m-1}( D D_0^{s/m}- D_0^{s/m} D).
  \end{eqnarray*}
  Substituting these above we obtain
  \[
  D_t = \frac sm\left(D_+^{s/m} D +D_0^{s/m}) D - D( D_+^{s/m} + D_0^{s/m})\right).
  \]
  Finally, $D(\gamma) = 0$ implies that \[D(\gamma_t) = -D_t (\gamma) = \frac sm D( D_+^{s/m} + D_0^{s/m})(\gamma).\]
  From here, the theorem follows as the kernel of $D$ does not include any invariant vector field.
\end{proof}

%

\end{document}